\documentclass[11pt]{article}
\usepackage{geometry} 
\usepackage[T1]{fontenc}
\usepackage[utf8]{inputenc}
\usepackage[english]{babel} 
\usepackage{amsthm} 
\usepackage{amsmath}
\usepackage{amssymb} 
\usepackage{mathtools,xparse} 
\usepackage{xargs}
\usepackage{bbm}
\usepackage{bbold}
\usepackage{hyperref}

\usepackage{algorithm}
\PassOptionsToPackage{noend}{algpseudocode}
\usepackage{algpseudocode}

\errorcontextlines\maxdimen

\makeatletter
    \newcommand*{\algrule}[1][\algorithmicindent]{\makebox[#1][l]{\hspace*{.5em}\thealgruleextra\vrule height \thealgruleheight depth \thealgruledepth}}%
\newcommand*{\thealgruleextra}{}
\newcommand*{\thealgruleheight}{.75\baselineskip}
\newcommand*{\thealgruledepth}{.25\baselineskip}

\newcount\ALG@printindent@tempcnta
\def\ALG@printindent{%
    \ifnum \theALG@nested>0
        \ifx\ALG@text\ALG@x@notext
        \else
            \unskip
            \addvspace{-1pt}
            \ALG@printindent@tempcnta=1
            \loop
                \algrule[\csname ALG@ind@\the\ALG@printindent@tempcnta\endcsname]%
                \advance \ALG@printindent@tempcnta 1
            \ifnum \ALG@printindent@tempcnta<\numexpr\theALG@nested+1\relax
            \repeat
        \fi
    \fi
    }%
\usepackage{etoolbox}
\patchcmd{\ALG@doentity}{\noindent\hskip\ALG@tlm}{\ALG@printindent}{}{\errmessage{failed to patch}}
\makeatother

\newbox\statebox
\newcommand{\myState}[1]{%
    \setbox\statebox=\vbox{#1}%
    \edef\thealgruleheight{\dimexpr \the\ht\statebox+1pt\relax}%
    \edef\thealgruledepth{\dimexpr \the\dp\statebox+1pt\relax}%
    \ifdim\thealgruleheight<.75\baselineskip
        \def\thealgruleheight{\dimexpr .75\baselineskip+1pt\relax}%
    \fi
    \ifdim\thealgruledepth<.25\baselineskip
        \def\thealgruledepth{\dimexpr .25\baselineskip+1pt\relax}%
    \fi
    \State #1%
    \def\thealgruleheight{\dimexpr .75\baselineskip+1pt\relax}%
    \def\thealgruledepth{\dimexpr .25\baselineskip+1pt\relax}%
}

\usepackage{color}
\usepackage{tikz}
\usetikzlibrary{shapes.geometric, arrows,automata} \usepackage{balance}
\usepackage{graphicx}
\usepackage[normalem]{ulem}
\usepackage{breqn}

\DeclareMathOperator*{\argmin}{arg\,min}

\DeclareMathOperator{\SAL}{ENERGY} 
\DeclareMathOperator{\SENDERS}{SENDERS}  
\DeclareMathOperator{\TIME}{TIME}     
\DeclareMathOperator{\SIM}{SIM} 
\DeclareMathOperator{\GP}{GP} 
\newcommand{\SAF}{\mathrm{SAF}}

\newcommand{\ignore}[1]{}

\newcommand{\polylog}{\mathrm{polylog}}
\newcommand{\Prob}[1]{\mathbf{P}\left(#1\right)}
\newcommand{\Otilde}{\tilde{O}}

\newcommand{\boldheading}[1]{\subsection*{#1}}

\newtheorem{theorem}{Theorem}[section]
\newtheorem{corollary}[theorem]{Corollary}
\newtheorem{lemma}[theorem]{Lemma}

\newtheorem{definition}[theorem]{Definition}

\newtheorem{proposition}[theorem]{Proposition}

\usepackage[utf8]{inputenc}

\newcommand{\ball}[3]{\mathrm{Ball}_{#1}(#2, #3)}

\newcommand{\dist}{\mathsf{dist}}

\newcommand{\poly}{\operatorname{poly}}

\begin{document}

\title{How to Wake Up Your Neighbors: Safe and Nearly Optimal 
Generic Energy Conservation in Radio Networks}

 \author{
 Varsha Dani\\
 {\small Rochester Institute of Technology}
 \and
 Thomas P. Hayes\\
 {\small University at Buffalo}
}




\date{}
\maketitle
\thispagestyle{empty}
\setcounter{page}{0}

\begin{abstract}
Recent work~\cite{chang2018energybroadcast,chang2020energyBFS,DGHP-LE-matchings} has shown that it is sometimes feasible to significantly reduce the energy usage of
some radio-network algorithms by adaptively powering down the radio receiver when it is not needed.
Although past work has focused on modifying specific network algorithms in this way, we now ask the question
of whether this problem can be solved in a generic way, treating the algorithm as a kind of black box.

We are able to answer this question in the affirmative, presenting a new general way to modify arbitrary
radio-network algorithms in an attempt to save energy.  At the expense of a small increase in 
the time complexity, we can provably reduce the energy usage to an extent that is provably nearly optimal within a certain class of general-purpose algorithms. 

As an application, we show that our algorithm reduces the energy cost of breadth-first search in radio networks from the previous best bound of $2^{O(\sqrt{\log n})}$ to $\mathrm{polylog}(n)$, where $n$ is the number of nodes in the network



A key ingredient in our algorithm is hierarchical clustering based on additive Voronoi decomposition done at multiple scales.  Similar clustering algorithms have been used in other recent work on energy-aware computation in radio networks, but we believe the specific approach presented here may be of independent interest.
\end{abstract}

\newpage

\section{Introduction}

For large networks of tiny sensors equipped with radio transceivers, the energy used in communication can be a bottleneck cost.
Although it has become standard practice to measure the cost of communication in terms of the number of messages or bits sent, 
it has been pointed out (see, for instance,~\cite{BarnesCMA10, PolastreSC05}) that the cost of using the receiver to listen for messages is often comparable
to, or even more than, the cost of transmission.  This becomes even more true as the size of the devices is scaled down.

This suggests that receivers should be turned off except when there is an actual message to receive.
However, this is easier said than done!  In general, effective economization of receiver use probably
requires redesigning our communication protocols from the ground up.  
As an example, previous work by Dani, Gupta, Hayes and Pettie~\cite{DGHP-LE-matchings}
on solving the Maximal Matching problem on a radio network with low energy expenditure, 
presented an algorithm that was very specific to its problem, in the sense that the algorithm is able to cleverly combine efficient message delivery with energy conservation and residual degree manipulation, all at once. 

Nevertheless, it is tempting to ask whether, in spite of this, some general-purpose technique can be used to 
minimize receiver usage when no messages are being sent nearby.  In the current work, we give a positive
answer to this question.  More precisely, we present a generic way to convert any given protocol 
into one that attempts to use its receiver more wisely, in a way tailored to the activity pattern of the
given algorithm.  
The amount of energy saved depends on properties of the algorithm being simulated.  
For some protocols, our construction will actually make them less energy efficient.
However, for at least a few protocols of interest, our technique improves the overall 
energy cost from $\poly(n)$ to $\polylog(n)$, where $n$ is the size of the network.
Moreover, there is a sense in which our algorithm is nearly the best possible, as we shall see.

This potential decrease in overall energy cost is not completely free.  Our algorithm incurs modest
(up to a polylog($n$) factor)
penalties in terms of overall running time, (sent) message complexity, and local computation.
These costs should be weighed against the possible benefits before deciding whether to deploy
our algorithm in a particular context.

We use a simple abstract model of distributed computation on Radio Networks~\cite{chlamtac1985broadcasting}, specifically the variant introduced in~\cite{chang2018energybroadcast}.  In this model, local computation is treated as free, 
and in each of a series of globally synchronized timesteps, 
each processor decides whether to SEND, LISTEN, or SLEEP.  Sending and Listening cost 1 unit of energy, 
but sleeping is free.  Messages can be up to $O(\log n)$ bits long.  In this setting, we try to simultaneously minimize
the total time complexity of our protocol, as well as the energy used by each node.

The motivation for our new construction comes from the structure of the simplest possible network algorithm: a naive
broadcast protocol, in which a message is repeatedly transmitted across the network at a rate of unit distance per timestep.
Considering the behavior of the ``active set,'' that is, the nodes that are actually transmitting or receiving a message
in a given timestep, we see that it behaves like a ripple moving across a pool of water, or the frontier of a parallel BFS algorithm, or even a light-cone propagating across spacetime. 

One aspect of this is that there is an upper limit on how fast the active set can travel, which potentially allows us to shut down nodes when the active set is known to be far away from them.  In particular, whenever a node's distance from the active set exceeds $d$, that node can afford to sleep for the next $d$ timesteps, because information cannot be transmitted faster than one edge per timestep.  Of course, to use this in practice, our hypothetical node needs advance warning of how far it is from the active set.



We achieve this by organizing nodes into clusters that spread information about the location of the active set faster than the
``speed of light'' in our simulated protocol.  Or, to put it more prosaically, at the cost of increased latency, our network is able
to transmit advance notice of the active set before it arrives.  To save as much energy as possible, we try to do this at all distance scales
simultaneously.  Thus, while nodes whose distance to the active set is only a few dozen may only know they can safely sleep for a dozen timesteps; at the same time, nodes whose distance to the active set is in the thousands may know they can safely sleep for hundreds of timesteps.  At a high level, this sounds very similar to the approach used by Chang, Dani, Hayes and Pettie~\cite{chang2020energyBFS}
for reducing the energy cost of BFS in the same model.  Indeed, the main difference in our current approach is in the way the clusters at different scales
are generated and communicate with each other.  However, as a result of these differences, we are able to exponentially improve on the energy usage of their BFS algorithm, and at the same time, apply this approach to a broad range of other problems.

We present a new general-purpose algorithm which we call SAF Simulation (see Algorithm~\ref{alg:SAF}) 
after its mnemonic, ``Sleep when Activity is Far.''
SAF Simulation accomplishes both the goal of building layered cluster decompositions, as well as using this structure to efficiently simulate a large class of algorithms. Our main result can be summarized (at a high level) as

\newcommand{\OPT}{\mathrm{OPT}}
\begin{theorem} \label{thm:main-vague}
For every Radio Network algorithm, $A$, SAF($A$) runs in time $O(\polylog(n) \TIME(A))$, and
with probability $1 - 1/\poly(n)$ produces the same output as $A$. 
Moreover, its energy cost satisfies, for every vertex $v$,
\[
\SAL(\SAF(A),v) \le \polylog(n) \OPT(A,v),
\]
where $\OPT$ is the least possible energy cost for a safe, synchronized, one-pass generic simulation algorithm.
\end{theorem}

A more precise (and technical) statement appears in Theorem~\ref{thm:SAF-main}. 
Some specific problems that can be solved by Algorithm~\ref{alg:SAF} for only $\polylog(n)$ per-node energy include:\ 
Broadcast, BFS, Leader Election, and Approximate Counting. 
Loosely speaking, the same can be said for any algorithm that does the bulk of its computation 
in $\polylog(n)$ BFS-like sweeps across the graph.  In all cases, this represents an exponential improvement
in energy usage, as compared to the naive implementation.

Our main tool for clustering our nodes is the ESTCluster algorithm of Miller, Peng and Xu~\cite{miller2013parallel}
(see also~\cite{miller2015improved}.)
This results in clusters that are the cells of a kind of generalized Voronoi diagram, as we discuss in Section~\ref{sec:AWVD}.
These clusters were used previously in the context of low energy computation in the same 
radio network model by Chang, Dani, Hayes and Pettie~\cite{chang2020energyBFS} for the problem of constructing breadth-first search.

In that work, they estimated long-range distances by recursively solving BFS on their level-1 cluster graph.
This has the apparent advantage that, in the resulting hierarchy of clusters, each cluster is a coarsening of
the previous lower-level cluster.  However, it has the disadvantage that the distortion of distances accumulates 
multiplicatively as one moves up levels.  This limits its usefulness for putting processors to sleep, since one
cannot risk missing the arrival of the active set.

In the current work, by contrast, instead of recursively constructing clusters of clusters, we instead run the ESTCluster algorithm repeatedly on the base graph, with different radius parameters, to generate our entire hierarchy of clusters.
An apparent disadvantage of this is that the resulting clusterings are not coarsenings and refinements of each other.
Fortunately, it turns out that we do not need this property.  Although it would be possible to 
use a sort of rounding to replace our clusterings with ones that are nested, there are algorithmic 
disadvantages to doing so.  Specifically, one nice thing about the generalized Voronoi cells created by the ESTCluster
algorithm is that they are star-shaped, so the entire cluster is connected by a BFS spanning tree rooted at the cluster center.
However, this property would be lost if redefined the clustering to be a coarsening of the lower-level clusterings.
This would appear to necessitate neighboring clusters to assist in what was previously intra-cluster communications.

\subsection{Other related work on energy complexity in radio networks}

A number of authors have studied energy complexity on \emph{single-hop} radio networks. Here the energy model is still that listening costs as much as sending, while sleeping and local computation are free, but the underlying graph is a clique. When nodes choose to transmit they are broadcasting their message to the entire graph, or at least anyone who is listening, and the main issue is just how to resolve contention for the channel. Moreover, in much of this work, there is no bound on message sizes. In this model, Nakano and Olariu~\cite{NakanoO00} gave an $O(\log\log n)$ energy algorithm for selecting distinct IDs in $\{1,\ldots,|V|=n\}$.
Bender, Kopelowitz, Pettie and Young~\cite{BenderKPY18} showed that $O(\log(\log^* n))$ energy suffices  for all devices to send messages if there is collision-detection available on the channel. 
Chang, Kopelowitz, Pettie, Wang and Zhan~\cite{ChangKPWZ17} showed this to be optimal, and studied the problems of Leader Election and Approximate Counting, both with and without randomization, with collision detection. They also studied tradeoffs between time, energy and error probability. Similar tradeoffs were also studied by Kardas et al.~\cite{KardasKP13}.
Jurdzinski, Kutylowski and Zatopianski~\cite{JurdzinskiKZ02b,JurdzinskiKZ02,JurdzinskiKZ02c,JurdziskiKZ03,JurdzinskiS02}.
studied the Leader Election and Approximate Counting problems in the absence of collision detection.

The single-hop notion of energy complexity was extended to arbitrary networks by Chang, Dani, Hayes, He, Li, and Pettie~\cite{chang2018energybroadcast} where they studied the Broadcast problem, with and without collision detection, and with and without randomization. The energy complexity was shown to be $\polylog(n)$ in all cases. However, their broadcast algorithm did not transmit messages along shortest paths. Later Chang, Dani, Hayes and Pettie~\cite{chang2020energyBFS} gave a $2^{O(\sqrt{\log n})}$ algorithm for BFS. They also showed that the diameter of a network is hard to approximate to within factor 2 using sublinear energy. 

Other models of energy cost have also been studied in the radio network literature. Gasnieniec et al.~\cite{GasieniecKKPS07}, 
Berenbrink et al.~\cite{BerenbrinkCH09} and Klonowski and Pajak~\cite{KlonowskiP18}, 
studied broadcast 
and gossiping problems under a cost model where the goal is to minimize the worst case number of transmissions per device. 
Klonowski and Sulkowska~\cite{KlonowskiS16} defined a distributed
model in which devices  can transmit messages at varying power levels, which can be chosen online. The devices in question are at random locations in the $d$-dimensional cube. 
A number of works~\cite{KutylowskiR03,KabarowskiKR06,GilbertKPPSY14,KingPSY18} also considered robustness of low-energy algorithms against a jamming attack. Here an adversary attempts to foil the devices' attempts to communicate by making noise on the channel. The goal here is not precisely low-energy use, but rather ``resource competitive'' energy use, where the processors combined cost for getting messages through should be commensurate with the adversary's budget for wreaking havoc.

\section{Preliminaries: the Low Energy Radio Network Model}

\boldheading{The Network}

We assume there is a communication network on an arbitrary, undirected connected graph $G = (V,E)$. 
At each node in $G$, there is a processor equipped with a transmitter and receiver to communicate with other nodes. 
There is an edge between nodes $u$ and $v$ in the graph if $u$ and $v$ are within transmission range of each other. 

The processors are identical, except for having unique IDs. They do not know the underlying graph $G$; indeed we assume that initially they do not even know their neighbors in the graph. A processor will become aware of (and remember) any neighbor once it has communicated with it. The processors do have a shared estimate on the size of the graph; that is, a number $n \ge |V|$ is known to all the processors and may be used in any algorithms they run. Accuracy of this estimate is not required for correctness of the algorithm, but the time and energy usage of algorithms will depend on it. Additionally we assume that the processors can generate independent streams of random bits. There is no shared randomness.  {\bf For randomized algorithms, we assume that each node does all of its coin flipping prior to the beginning of the algorithm, generating a string of random bits to be read off and used at the appropriate time in the algorithm.} This is a standard construct for viewing a randomized algorithm as a deterministic algorithm with an additional random input.

\boldheading{Time}

Time is divided into discrete, synchronous timesteps, and the processors agree on a time $t=0.$ 
In each timestep a processor can choose to do one of three actions: SEND, LISTEN, or SLEEP. When a processor decides to SEND at time $t$, it also chooses a message of size $O(\log n)$ bits, which may potentially be received by a subset of its neighbors who happen to be listening at time $t$. We say `potentially' because there are a number of different models for what happens if a listening node encounters a collision, \emph{i.e.} two or more of its neighbors are broadcasting messages in the timestep that it is listening. Nevertheless, at a minimum, we can say that the chosen message travels from a node $u$ to a neighbor $v$ of $u$ at time $t$ if 
\begin{itemize}
    \item $u$ decides to SEND at time $t$, 
    \item $v$ decides to LISTEN at time $t$ and 
    \item no other neighbor of $v$ decides to SEND at time $t$.
\end{itemize}

\boldheading{Collisions and Message Delivery}

There are several different models for how to handle collisions, \emph{i.e.,} the situation where a node listens while two or more of its neighbors are transmitting messages. 
In the most permissive of these, the LOCAL model, $v$ receives all the messages sent by its neighbors (moreover, in this model messages are not bounded in size.) As already specified, we are not working in this model. 

A more restrictive model is the {\bf Collision Detection model (CD)} where, when a listening node does not receive a message, it can can tell the difference between silence (no neighbors sending) and a collision (more than one neighbor sending).  

Another model of interest is the {\bf No Collision Detection model (no-CD)}, which is even more restrictive: here, collisions between two or more messages are indistinguishable from silence.

Prior work~\cite{chang2018energymultihop,chang2018energybroadcast,chang2020energyBFS} used exponential backoff to deal with collisions in both the CD and no-CD models, making the distinction between these models less important, except in the case of deterministic algorithms. The idea in exponential 
backoff is that each timestep is simulated using $O(\log^2 n)$ timesteps: $O(\log n)$ rounds of $O(\log n)$ timesteps each. A listener listens for the entire interval of length $O(\log^2 n)$ or until a message is received. In each round, senders flip coins at each timestep to decide whether to (re-)transmit, or to retire for the rest of the round, resulting in a constant probability that a listener gets a message from one exactly of the senders. The net result of this $O(\log^2 n)$-time backoff procedure is that, with high probability, every listener receives a message from at least one of its sending neighbors.


The fact that collisions can be handled in this manner inspires the definition of the following alternative message delivery model without collisions.
In the ``OR'' model of message delivery, every time a node listens, it receives an arbitrary message sent by one of its neighbors
in that timestep.  The only exception is when no neighbor of the listening node chose to send in that timestep, in which case
no message is received.




For convenience, we will assume for the rest of this paper that we are working in the OR model of message delivery.

\boldheading{Energy Usage}

We measure the cost of our algorithms in terms of their energy usage. We assume that a node incurs a cost of 1 energy unit each time that it decides to send or listen. When the node is sleeping there is no energy cost. We also assume that local computation is free. The goal of energy aware computation is to design algorithms where the nodes can schedule sleep and communication times so that the energy expenditure is small, ideally $\polylog (n)$, without compromising the time complexity too much, i.e., the running time remains polynomial in $n$.

\section{Cluster Graphs and Distance Approximation}
In this section we introduce our framework for getting distance estimates that will be used by the nodes in the simulation algorithm.

\subsection{Graph theoretic preliminaries}

Let $G=(V,E)$ be a graph. Let $d:G\times G \rightarrow \mathbb{N}$ be the shortest path metric on $G$. 
Consider a partition 
\[\mathcal{P} = (V_1, V_2, \dots, V_k)\] of $V$ into pairwise disjoint sets of vertices. We will call each $V_i$ a "cluster".  The \emph{cluster graph}, also called the \emph{quotient graph} and denoted $G/\mathcal{P}$, is a graph whose vertex set is the set of clusters, 
$\{V_1, \dots V_k \}$. There is an edge in the cluster graph between $V_i$ and $V_j$ if there are vertices $u \in V_i$ and $v \in V_j$ such that $(u,v) \in E$.

Although the definition of a quotient graph does not require either $G$ or the subgraphs induced by the clusters $V_i$ to be connected, in our work we will assume that both are the case. It is easy to see that since $G$ is connected, so is $G/\mathcal{P}.$
\begin{figure*}
\begin{center}
\begin{tikzpicture}
\tikzstyle{every node} = [draw, circle]
\node (a) at (1,1){a};
\node (b) at (1,3){a};
\node (c) at (3,2){a};
\draw (a)--(b)--(c)--(a);;
\node (l) at (3,4){a};
\draw (b)--(l);
\draw (c)--(l);
\node (e) at (5,2){b};
\node  (m) at (5,4){b};
\node (f) at (6,3){b};
\draw [dotted, thick] (l)--(m);
\node (s) at (3.5,6){b};
\draw (m)--(s);
\node  (z) at (6,6){c};
\draw [dotted, thick](z)--(s);
\draw (m) -- (e);
\draw [dotted, thick] (a)--(e);
\draw (m) -- (f)--(e);
\node (g) at (8,2){d};
\node  (x) at (8,5){d};
\draw [dotted, thick] (x)--(z);
\draw (x) -- (g);
\draw [dotted, thick] (f)--(g);
\draw [dotted, thick] (m)--(x);

\node (A) at (11,2){A};
\node (C) at (13.5,6){C};
\node (B) at (13,4){B};
\node (D) at (16,3){D};
\draw (A)--(B)--(C);
\draw (C)--(D)--(B);
\end{tikzpicture}
\end{center}
\caption{At left, a graph, $G$, whose vertex set has been partitioned into four clusters, each labelled by an element of $\{a,b,c,d\}$.  At right, the corresponding quotient graph. Edges of $G$ between different clusters are indicated with dotted lines. These become the edges of the quotient graph, after the elimination of duplicates.}
\label{fig:quotient-graph}
\end{figure*}
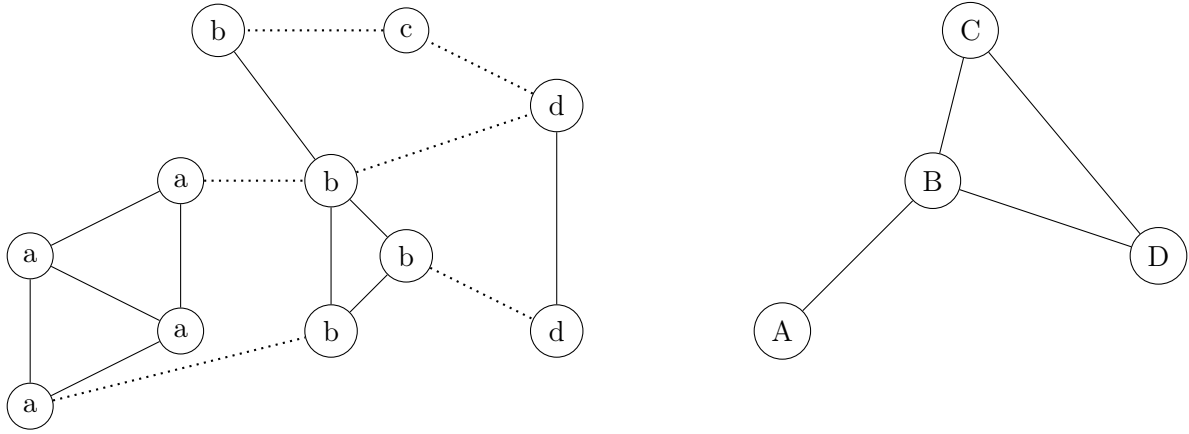

Given a partition $\mathcal{P}$, we will denote the cluster containing vertex $u$ by $[u].$
Since the cluster graph is a graph, we can also define the shortest path distance metric on the cluster graph. We are interested in partitions where the individual clusters have comparable diameters, and the distances between clusters in the cluster graph are (approximately) scaled versions of the distances in the underlying graph.  

\subsection{Approximately Distance-Preserving Partitions}

\begin{definition}
Let $G = (V,E)$ be a graph, and let $R, \alpha, \beta \ge 1$.
Let $\mathcal{P}$ be a partition of $V$. Let $d$ denote the distance metric in $G$ and $d^*$ denote the distance in $G/\mathcal{P}$.
We say $\mathcal{P}$ is an $(R, \alpha, \beta)$-approximately distance-preserving partition if
\begin{itemize} 
\item for every 
pair of nodes $u$ and $v$ with $d(u,v) \le R$, we have $d^*([u],[v]) \le \alpha $
\item for every pair of nodes $u$ and $v$  with $d^*([u], [v])=0$, (that is, $u$ and $v$ are in the same cluster)  we have
$d(u,v) \le \beta R.$
\end{itemize}
\end{definition}

Taken together, the two halves of the above definition ensure that up to a multiplicative factor
in the range $[1/\alpha, \beta]$, as well as possible rounding issues, 
we have, for all $u,v \in V$, that $d^*([u],[v]) \approx \frac{d(u,v)}{R}$.
Specifically,

\begin{lemma}\label{lem:adpp}
Let $G = (V,E)$ be a graph, and $\mathcal{P}$ be an $(R, \alpha, \beta)$-approximately distance-preserving partition of $V$.
Let $d$ and $d^*$ denote the shortest path distance metrics in $G$ and $G/\mathcal{P}$ respectively. Then, for all $u, v \in V$,
\[
\left\lfloor\frac{d(u,v)}{\beta R +1}\right\rfloor \le d^*([u],[v]) \le \alpha \left\lceil\frac{d(u,v)}{R}\right\rceil
\]
\end{lemma}

\begin{proof}
Let $u, v \in V$. Let $d^*([u], [v])=k$ and let $W_0, W_1, \dots, W_k \in V(G/\mathcal{P})$ be such that $W_0 =[u], W_k =[v]$ and $W_0, W_1, \dots, W_k$ is a shortest path from $[u]$ to $[v]$ in $G/\mathcal{P}$. Then, for $0\le i < k$, since $W_i$ and $W_{i+1}$ are adjacent in $G/\mathcal{P}$, there are vertices $y_i \in W_i,  x_{i+1} \in W_{i+1}$ such that $(x_{i+1}, y_i) \in E(G)$. Let $x_0 =u$ and $y_k=v$.
For all $i$, since $x_{i}$ and $y_i$ are in the same cluster, $d(x_{i}, y_i) < \beta R$, and if we connect these pairs by shortest paths, we have constructed a path of length at most $k + (k+1) \beta R$ from $u$ to $v$ in $G$. It follows that 
\[
d(u,v) \le k + (k+1) \beta R = k(\beta R +1) + \beta R
\]
Since $k =d^*([u], [v]) $, rearranging terms gives us the first inequality.

Now suppose $d(u,v) = qR + r$, where $r<R$. Consider a shortest path from $u$ to $v$ in $G$, and let $w_i$ be the vertices at distance $iR$ from $u$ on the path. The path has thus been cut into $q' = q+ \lfloor r/R\rfloor = \lceil d(u,v)/R \rceil$ pieces each of length at most $R$. 
Let $w_{q'}=v$. By the first property, $d^*([w_i],[w_{i+1}]) \le \alpha$. Therefore, by the triangle inequality, 
\[
d^*([u],[v]) \le \alpha \left\lceil \frac{d(u,v)}{R} \right\rceil
\]
which completes the proof.
\end{proof}


A nice example of an approximately distance-preserving partition is for the square grid.  Let $G$ be the $n \times n$ square
grid graph, and suppose $R$ is a divisor of $n$.  We partition $V$ by rounding each point $(x,y)$ down to the nearest multiple of $R$,
$\left(R \lfloor \frac{x}{R} \rfloor, R \lfloor \frac{y}{R} \rfloor\right)$, and placing two vertices in the same cluster if they
round to the same multiple of $R$.  It is easy to see that this example is an $(R, 2, 2)$-approximately distance-preserving partition.  Quite similar constructions can be done
for all real $R \ge 1$, and for many other ``homogeneous'' graphs, such as lattice graphs of fixed dimension.

Although it may not be immediately obvious, such approximately distance-preserving partitions exist for all graphs, for all $R \ge 1$,
and for $\alpha, \beta = O(\log n)$, as we shall see next.  

\subsection{Additive Weights Voronoi Diagrams and the MPX Algorithm}\label{sec:AWVD}

Additively Weighted Voronoi Decomposition (AWVD) (see, e.g., Phillips~\cite{phillips2014tessellation}),
also known as hyperbolic Dirichlet tesselation, is a well-studied concept for real Euclidean domains.
We start with a finite set of points in the plane, called generators, each of which is assigned a real-valued weight.
Each point $x$ in the plane is assigned to the generator $g$ that minimizes $\|x-g\| - W(g)$.  After discarding
any empty cells, this defines a partition of the plane into a finite collection of cells.  When the weights are
all zero, this corresponds to the usual notion of Voronoi diagram,  and the boundaries of the cells are line segments
and rays.  For general weights, the boundaries of the cells are hyperbolic arcs.  The cells are star-shaped with
respect to their generators, but are generally not all convex.

For a finite graph, $G = (V,E)$, the analogous concept is a partition of $V$ based on assigning a
real-valued weight $W(v)$ to each vertex $v \in V$.  We say that vertex $u$ belongs to the cell
generated by vertex $v$ if $v = \argmin_{v' \in V} d(u,v') - W(v')$.  For convenience, we will assume that
no two weights $W(v), W(v')$ differ by an integer, so that the cells are defined unambiguously.
Two vertices belonging to the same cell is an equivalence relation, so each AWVD gives rise to 
a corresponding cluster graph.

Miller, Peng, and Xu~\cite{miller2013parallel} proposed a simple randomized graph-partitioning algorithm to obtain a decomposition with certain properties, specifically that the clusters have small diameter and only a small fraction of the edges of the graph are cut.
In their construction, starting with a (common) parameter $R$, each vertex $v$ independently samples a random variable
$\delta_v \sim \operatorname{Exponential}(1/R)$ from the exponential distribution with mean $R$. A cluster starts forming at each vertex $v$ at time $-\delta_v$, and spreading through the graph at a uniform rate of
one edge per time unit. Each vertex $u$ either joins the first cluster to reach it before time $-\delta_u$ or starts its own cluster at time $-\delta_u$ if no other cluster has recruited it before that time. 
Haeupler and Wajc~\cite{haeupler2016faster} showed that with minor modifications, this algorithm can be efficiently implemented in the Radio Network model.

We note that the MPX decomposition is, in fact, an AWVD where the weights $W(v)$ are 
independent exponentially-distributed random variables with mean $R$. We will next see that this decomposition has some good distance approximating properties.

As shown in \cite{miller2013parallel}, the clusters have diameter $O(R\log n)$. We state this slightly more precisely:

\begin{lemma}\label{lem:mpx-adpp-beta}
With probability at least $1-\tfrac{1}{n^2}$, the clusters in the MPX decomposition with parameter $R$ have diameter at most $3R\log n$.
\end{lemma}

In fact this is easy to see. If $X$ is an exponentially distributed random variable with mean $R$ then $\Prob{X>t} = e^{-t/R}$. Thus, the probability that $X > 3 R\log n$ is at most $1/n^3$. By a union bound, it follows that with probability at least $1-\tfrac{1}{n^2}$ the first cluster starts forming \emph{after} time $-3R\log n$ and since each vertex either joins a cluster or starts its own by time $t=0$, the cluster diameters are at most $3R\log n.$

Furthermore, Chang \emph{et. al}~\cite{chang2020energyBFS} showed that not too many clusters are close to a single vertex. Specifically, if $\mathcal{P}$ is the partition determined by the MPX algorithm with parameter $R$, and $G^* = G/\mathcal{P}$ is the corresponding cluster graph, then Lemma 2.1 from ~\cite{chang2020energyBFS} (translated into our notation) can be stated as follows:

\begin{lemma} \label{lem:cluster-exponential-hits-in-ball}
For every positive integer $j$ and $\ell>0$, 
the probability that the number of $G^*$-clusters 
intersecting $\ball{G}{v}{\ell}$
is more than $j$ is at most
\[
\big(1 - \exp(-2 \ell/R)\big)^j.
\]
\end{lemma}

\noindent
In particular, setting $\ell = R$ we see that for any $v \in V$, the probability that there are more than $C\log n$ clusters intersecting the ball of radius  $R$ around $v$ is at most 
\[
\left(1- \frac{1}{e^2}\right)^{C\log n} \le \exp\left(\frac{-C \log n}{e^2}\right).
\]
Setting $C = 23 > 3 e^2$, and taking a union bound over all $n$ choices for $v$, it follows that 
\begin{corollary}\label{cor:mpx-adpp-alpha}
\[
\Prob{\mbox{for all $v \in V$, at most } 23 \log n \mbox{ clusters intersect } \ball{G}{v}{\ell}} \ge 1- \frac{1}{n^2}
\]
\end{corollary}

\noindent
Combining Lemma~\ref{lem:mpx-adpp-beta} and Corollary~\ref{cor:mpx-adpp-alpha}, we have shown that 

\begin{proposition}
With probability at least $1 - \frac{2}{n^2}$,
the partition $\mathcal{P}$ produced by the MPX algorithm with parameter $R$ is a $(R, 23\log n, 3\log n)$-approximately distance-preserving partition  
\end{proposition}


\subsection{Multi-Scale Clustering}

\begin{definition}
Suppose, for $1 \le i \le \ell$, we have an approximately distance-preserving partition $\mathcal{P}_i$ with parameters
$(R_i, \alpha_i, \beta_i)$, where $R_1 < \dots < R_\ell$.
By convention, we extend this definition to the case $i=0$ by setting $R_0 = \alpha_0 = \beta_0 = 1$,
and letting the $i=0$ partition be the partition of $V$ into singleton vertex sets.
Then we say we have a multi-scale clustering with parameters $(\mathbf{R}, \boldsymbol\alpha, \boldsymbol\beta).$
\end{definition}

\begin{definition}
Suppose, for $1 \le i \le \ell$, we have an approximately distance-preserving partition $\mathcal{P}_i$ with parameters
$(R_i, \alpha_i, \beta_i)$, where $R_1 < \dots < R_\ell$.
By convention, we extend this definition to the case $i=0$ by setting $R_0 = \alpha_0 = \beta_0 = 1$,
and letting the $i=0$ partition be the partition of $V$ into singleton vertex sets.
Suppose further that for $1 \le i \le \ell$, we have
\begin{equation}\label{eqn:multi}
    \frac{R_{j+1}}{R_j} \ge (2\alpha_j +1) (\beta_j +1) -1
\end{equation}

Then we say we have a multi-scale clustering with parameters $(\mathbf{R}, \boldsymbol\alpha, \boldsymbol\beta).$
\end{definition}

Notice that there is no requirement that the lower-level clusters be refinements of the 
higher-level clusters.  Despite this, as we will see, multi-scale clusterings have a useful, albeit weaker, nesting property that controls the relationship between clusters at consecutive scales.

Suppose $\mathcal{P}_1, \mathcal{P}_2, \dots \mathcal{P}_{\ell}$ is a multi-scale clustering on $G$, with parameters $(\mathbf{R}, \boldsymbol\alpha, \boldsymbol\beta)$. Let $G_1, \dots G_{\ell}$ be the corresponding cluster graphs, and $d_1, \dots, d_{\ell}$ the corresponding distance metrics. For $v \in V$, let $[v]_j$ denote the cluster containing $v$ in $\mathcal{P}_j$. 
As usual, $\ball{G_j}{[v]_j}{r}$ denotes the ball of radius $r$ in $G_j$. We will also need a notation for the vertices in the underlying graph $G$, whose clusters belong to this ball. To this end we define
\[
B_j(v, r) = \{ w \in V \, | \, d_j( [w]_j, [v]_j) \le r\}
\]
The following nesting property is true for all vertices, at all scales.
\begin{lemma}\label{lem:nesting}
Let $1\le j < \ell$, \, $u \in V$, and $v \in \ball{G}{u}{R_{j+1} - R_{j}}$. That is, $d(u,v) \le R_{j+1} - R_{j}$. Then
\[
B_{j}(v, 2\alpha_{j}) \subset  B_{j+1}(u, 2\alpha_{j+1})
\]
\end{lemma}
\begin{proof} 
Let $w \in B_{j}(v, 2\alpha_{j})$. Then, by definition, $d_{j}([v]_{j} [w]_{j}) \le 2 \alpha_{j}$.
Since $\mathcal{P}_{j}$ is an $(R_j, \alpha_j, \beta_j)$ approximate distance-preserving partition, by Lemma~\ref{lem:adpp},
\[
\left\lfloor\frac{d(v, w)}{\beta_j R_j +1}\right\rfloor \le 2\alpha_j  
\]
Removing the floor and rearranging the terms, we get 
\begin{align*}
    d(v, w) &\le  (\beta_j R_j +1)(2\alpha_j +1)\\
    &\le (2\alpha_j +1)(\beta_j +1)R_j \\
    &\le R_{j+1} + R_j
\end{align*}
where the last line follows from Equation~\eqref{eqn:multi}.
By the triangle inequality, 
\[
d(u, w) \le d(u,v)+d(v,w) \le R_{j+1}-R_{j} + R_{j+1} +R_j = 2R_{j+1}
\]
Applying Lemma~\ref{lem:adpp} again, we have 
\[
d_{j+1}([u]_{j+1}, [v]_{j+1}) \le \alpha_{j+1} \left\lceil \frac{d(u,v)}{R_{j+1}} \right\rceil \le 2\alpha_{j+1} 
\]
so that $w\in B_{j+1}(u, 2\alpha_{j+1})$, completing the proof.
\end{proof}

Although the only restriction we have imposed so far of the parameters $(\mathbf{R}, \boldsymbol\alpha, \boldsymbol\beta)$ so far are that the $R_j$s grow at a certain rate relative to the $\alpha_j$s and $\beta_j$s, an interesting special case is when they grow geometrically, each level being a scaled up version of the former. 

\begin{definition}
In the special case when, for $1 \le j \le \ell$, $R_j = R^j$, for some $R > 12$, and 
$\alpha_j = \beta_j = \lfloor\sqrt{R}/2\rfloor$, we call this a geometric multi-scale clustering with parameter $R$.
\end{definition}

We saw earlier that the MPX clustering with parameter $R$ is, \emph{w.h.p,} a $(R, 23\log n, 3\log n)$-approximate distance-preserving partition. Naturally, it is also a $(R, 23\log n, 12\log n)$ approximate distance-preserving partition. Setting $\alpha = \beta = 23 \log n$ and $R = \Theta(\log^2 n)$ we can get a geometric multi-scale clustering by constructing MPX partitions for all scales $R, R^2, R^3 \dots$

\subsection{Simulating cluster-graph algorithms on the underlying graph}

Suppose we already have approximately distance-preserving partitions at all scales
from $1$ up to the diameter of $G$.  The ability to run graph algorithms on the corresponding
quotient graphs will be a rather useful primitive to add to our toolbox.  Because the distances
in these graphs are scaled down by a large factor, we may reasonably expect that they can be run at a much
lower cost.  But what is the cost of simulating these algorithms on the actual network?

For the most part, this question was already answered by Chang, Dani, Hayes and Pettie~\cite{chang2018energybroadcast}.
We briefly describe the approach used to simulate one timestep of computation on the cluster graph.
Each node within one cluster in our partition uses part of its memory to record the state of
a hypothetical processor corresponding to the cluster.  At the beginning and end of the simulated timestep, 
we require that this state be the same for every node in the cluster.  Depending on whether the cluster
state indicates we should Send, Listen, or Sleep, every node in the cluster does this (INTERCAST).  Next,
if the operation was Listen, all the nodes that received a message propagate these up towards the cluster
center (UPCAST).  Since every node in the cluster knows its distance from the root, this propagation can
be synchronized so that each node only needs to Listen once, in the same timestep that its children might Send.
Since we are in the OR model, we only require that, from among the messages received, an arbitrary one is received.
Next, the cluster center updates its state based on the received message, and broadcasts the result within the 
cluster (DOWNCAST), after which the simulation of one computational step on the cluster graph is complete.
Since the messages are being broadcast, rather than sent along edges, there are some subtleties in the
details.  For instance, even though we have described the algorithm under the assumption that message delivery
in both the underlying graph, and in the cluster graph, takes place using the collision-free OR model,
there is still a need for backoff, to prevent messages accidentally crossing between adjacent clusters during
UPCAST and DOWNCAST.  However, for purposes of the present work, these details are unimportant.

Since each of the three stages, INTERCAST, UPCAST, and DOWNCAST, costs $\Otilde(1)$ energy per node in the cluster,
the per-node energy cost for running an algorithm on the cluster graph is within a polylog factor of the per-node
energy cost for simulating it on the underlying graph.  When the cluster ``radius'' parameter is $R$, the 
latency for the UPCAST and DOWNCAST is $\Otilde(R)$, and so this becomes a multiplicative factor for the time
complexity of the simulation.  

\begin{algorithm} \caption{High-level description of Upcast, Downcast, and Intercast algorithms, from~\cite{chang2018energybroadcast}. }
\begin{algorithmic}
\State{ }
\Comment{For each of the subroutines below, we require that either all the nodes in a particular cluster participate, or none do.}
\Procedure{Upcast}{$j, m$} \Comment{ $j$: current cluster height, $m$: message to send}
\State{} \Comment{ Guarantee: if any node in the cluster participates with a non-null message, then
one of these messages is received by the cluster center, who then stores it in their message variable.}
\EndProcedure
\Procedure{Downcast}{$j, m$} \Comment{ $j$: current cluster height, $m$: message to send}
\State{} \Comment{ Guarantee: if the cluster center participates with a non-null message, then
this message is received by each node in the cluster, who then stores it in their message variable.}
\EndProcedure
\Procedure{Intercast}{$j, m$} \Comment{ $j$: current cluster height, $m$: message to send}
\State{} \Comment{ Guarantee: if at least one neighboring node is contained in a participating level-$j$ cluster that has a non-null message, then this node receives such a message, and store it in their message variable.}
\EndProcedure

\end{algorithmic}
\end{algorithm}

\section{Simulating Radio Network Algorithms}

Let $A$ be a radio network algorithm. When talking about the time and message complexity of $A$, since one does not charge for listening, one usually only specifies which timesteps $A$ requires a node to SEND, with the implicit assumption that whenever a node is not SENDing it is LISTENing.  
Given any algorithm in this form, it is trivial to naively convert it to what we view here as the standard form,
by making all LISTENs explicit.
After such a naive conversion, the algorithm will have per-node energy cost equal to its time complexity, which is
not good, unless the time complexity is already very small.  Our plan for improved energy efficiency is now to simulate this explicitly wasteful algorithm
by one in which many of the LISTENs have been replaced by SLEEPs without compromising the correctness of the algorithm, and with a relatively small overhead in time complexity.

Whenever we talk about simulating a radio network algorithm in this paper, we mean replacing it with another
equivalent algorithm which is moreover \emph{safe black-box} and \emph{synchronized one-pass}, terms which we are about to define.
By equivalent, we mean that at the end of the simulation, each node, $v$, in the network will have computed its view
of a correct transcript of the original algorithm; that is, a record of all messages $v$ sent or received at each timestep in the original algorithm, as well as the correct final internal state.
By \emph{black-box}, we mean that when the simulating protocol wants to perform a step of the simulated algorithm, it does so by invoking an oracle that tells it how to update its state, whether to SEND, LISTEN, or SLEEP, and what messages to send.  By \emph{safe}, we mean that, the simulating protocol never causes a node to SLEEP during a round in which the original protocol would SEND or receive a message from a SENDing neighbor, regardless of what oracle is provided.
Specifically, we don't just mean that the simulating algorithm performs correctly when simulating $A$; it must avoid SLEEPing incorrectly for all possible simulated algorithms.
Finally, by \emph{synchronized one-pass}, we mean that, at certain designated timesteps $\tau(t)$, all nodes either SLEEP or simulate step $t$ of the simulated algorithm.  Each timestep is simulated only once, and the timesteps $\tau(t)$ are a strictly increasing function of $t$.  It so happens that for our particular simulation algorithm, the timesteps $\tau(t)$ are a fixed function of $t$, known in advance, rather than determined adaptively, but this will not be important for our analysis.

\subsection{Characterizing Optimal Simulation}

In order to come closer to getting our hands on what it would mean for a safe simulation algorithm to be optimal
or nearly optimal, we introduce the following generous abstract model for what it needs to do. 

Our notion of an generic algorithm simulator is a radio network algorithm, where, at some timesteps, it
invokes a black-box simulation of one timestep of the simulated algorithm, $A$, and at other timesteps, it 
does its own thing, which may involve metadata gathered from the usage pattern of $A$.  Since $A$ is arbitrary
and unknown, we cannot count on knowing the meaning of the messages sent by $A$, but there may still be useful
information to be gleaned by information about which nodes of $A$ are SENDing, LISTENing, or SLEEPing at a given timestep.  

With the above in mind, we define the Generic Simulation Model as follows.  For each timestep, at each node, our
algorithm must choose to either SEND, LISTEN, SLEEP, or SIMULATE.  Steps when our algorithm SENDs, LISTENs or SLEEPS work the same way as in the radio network model.  On a SIMULATE step, a black-box for the simulated algorithm, $A$, simulates one timestep, which may involve SENDING, LISTENING, or SLEEPing.  Each node $v$ is given the following information:
\begin{enumerate}
    \item whether $v$ chose to SEND, LISTEN, or SLEEP in the simulated step.
    \item whether any message was received by $v$ in the simulated step.
    \item the identity of the next timestep, $t_v$, at which $A$ would make $v$ SEND, under the assumption that no further messages were received by $v$ before $t_v$.
\end{enumerate}

At the beginning of the algorithm, we further suppose that the simulator at each node $v$ is given to know the first timestep $t_v$, at which $A$ would make $v$ SEND, under the assumption that no messages were received by $v$ before $t_v$.  We note that, if the simulated algorithm $A$ were fully event-driven and deterministic, the ``unprovoked next-send times" $t_v$ would all be $+\infty$, except for those nodes that send in the first timestep.  However, for general algorithms, and in particular, for randomized algorithms, such as the MPX clustering algorithm, there may be many times at which new sequences of SEND steps start without being directly preceded by an incoming message.
For randomized algorithms, since we have assumed that each node does all of its coin flipping prior to the beginning of the algorithm, note that the times $t_v$ can be accurately predicted in advance. Since the definition of $t_v$ assumes that no new external information reaches node $v$ between the current step and timestep $t_v$, the information needed for this calculation is always available to the simulator at the current timestep.  Here, we are taking full advantage of the standard assumption that local computation is free; however, we also believe that in most settings, calculating  $t_v$ at time $t$ can be done much more quickly than the ultra-naive approach of doing $t_v-t$ steps of black-box simulation of $A$.

\newcommand{\tnext}{t_{\mathrm{next}}}
Next, we define a Psychic Synchronized Generic Simulation Model, as follows.  
Everything about it is the same as for the
Generic Simulation Model, above, except that, in addition to being informed about messages received by $v$ in
the preceding time step, in this model, we assume that the simulator starts out knowing the entire graph, including node IDs, and is informed, at each time step, of the entire history of metadata at all nodes, including the information about times at which the current node was asleep.  In particular, this includes knowing exactly which nodes SENT, LISTENed, and SLEPT at all times $\le t$, as well as all timesteps $\tnext(w,t)$ at which future messages are scheduled to be sent
in the absence of provocation, for all vertices $w \in V$.

The other crucial assumption we make in the PSGSM is that all nodes advance their simulation of $A$ in a synchronized way.  That is, every node simulates the behavior of $A$ at time $t$ at the same time.  This assumption seems natural, considering that the success or failure of the message deliveries depends on whether collisions occur, but it is still an assumption.  Now, considering that in the PSGSM model, all nodes receive, for free, the entire history of relevant metadata for the entire graph, at each timestep they are awake, there seems to be no point in further communication between nodes, except for when steps of $A$ are being simulated.  Thus, we find that, in the PSGSM model, the simulated algorithm $A$ should always run in real time.

Now, in both of the above models, we have at least one expectation of a simulation algorithm, and that is \emph{correctness}: at the end of the simulation, we want every node to have its local view of the transcript of the actual computation done by $A$ from the given inputs.  In the case of randomized algorithms, we view each node's pre-flipped coins as part of its input, which allows us to reduce to the case of a deterministic algorithm in the usual way.

Our motivation in introducing the Psychic Synchronized (PSGS) model is that, in this model, we can precisely characterize the optimal energy usage of any correct simulation in terms of the behavior of a simple greedy algorithm,
Algorithm~\ref{alg:greedy-psychic}.

\begin{algorithm} \caption{The Greedy Psychic Algorithm: conserve energy while simulating a given Radio Network agorithm in the PSGSM model.}
\label{alg:greedy-psychic}
\begin{algorithmic}
\Procedure{$\GP$}{$A$} \Comment{ $A$: the simulated algorithm }
\myState{$t^* \gets 1$}
\For{$t \gets 1$ to $T$}
\If{$t < t^*$}
\myState{SLEEP}
\Else{}
\myState{SIMULATE timestep $t$ of $A$} 
\myState{$t^* \gets \tnext(v,t)$} \Comment{our next scheduled SEND}
\For{every vertex $w \ne v$}
\myState{psychically receive $\tnext(w,t)$} \Comment{time of $w$'s next scheduled SEND}
\myState{$t^* \gets \min\{ t^*, \tnext(w,t) + (\dist(v,w) - 1) \}$}
\EndFor
\EndIf
\Comment{SLEEP until time $t^*$}
\EndFor
\EndProcedure
\end{algorithmic}
\end{algorithm}

  
  \begin{theorem} \label{thm:greedy-psychic}
    This "Greedy Psychic" algorithm is optimal among all correct simulation algorithms in the PSGSM model,
  in the sense that, for every simulation algorithm SIM in this model, either there exists a radio network 
  algorithm $A$ and an input $x$ such that SIM($A,x$) makes a mistake (some node does not end with the correct
  transcript), or, for every algorithm $A$, input $x$, and vertex $v$, we have 
  $\SAL(\GP,A,x,v) \le \SAL(\SIM,A, x, v)$.
  \end{theorem}
  
  \begin{proof}
  The proof is by induction of the number of computational steps.
  Suppose we have fixed $A,v,x$, and that $\SIM$ is known to be an always-correct simulator.  Our inductive
  hypothesis is that, by the times the cumulative energy costs of $\GP(A,v,x)$ and $\SIM(A,v,x)$ both
  reach $i$, the greedy simulation will have completed simulating at least as many timesteps of $A$
as $\SIM$ will have.  Assume the hypothesis for $i-1$.  Then $\SIM$ wakes up for the $i$'th time at or before
Greedy does,  say at time $t$.  
At this point it can go back to sleep, but the latest it can sleep is $t^*(v,t)$, defined as above, since
$\SIM$ needs to be safe.  Greedy wakes up for the $i$'th time at some time $t' \ge t$, and goes to sleep until
time exactly $t^*(v,t')$.  Now, the crucial point is that the function $t^*(v,\cdot)$ is an increasing function of its second argument, which is clear from its intention, but also from the facts that $\tnext(w,\cdot)$ can only change
when $w$ receives a message from a neighbor, which implies that $t^*(w,t) - t \le 1$, and hence that
$\dist(w, \SENDERS)$ never decreases by more than 1 in a single timestep.
Hence $t^*(v,t') \ge t^*(v,t) \ge $ whenever $\SIM$ decides to wake back up, which completes the inductive step.
 \end{proof}
  
The main consequence of Theorem~\ref{thm:greedy-psychic} will be that Algorithm SAF
is within a polylog($n$) multiplicative factor of optimal, at least among algorithms that perform a single
synchronized simulation of the simulated algorithm $A$. This near-optimality holds in a vertex-by-vertex
and algorithm-by-algorithm manner.  We point out that without the assumption of synchronicity, we cannot
expect to get guarantees of this kind.  For example, if we want to make an asynchronous algorithm that                   minimizes the listening cost for a specific vertex $v$, we can make $v$ wake up very infrequently, and
ensure that its neighbors are always ``holding $v$'s messages ready for it."  In this way, the energy complexity
for $v$ can be reduced essentially to the number of messages that $v$ needs to send or receive; this favoritism towards $v$ would presumably be more than paid for by increased costs incurred at other vertices.
This strongly suggests that, without the assumption of synchronous one-pass simulation, there may not exist a single simulation algorithm that simultaneously minimizes the energy costs for all nodes.

\section{The Simulation Algorithm} 
 
In this section we describe our recursive simulation algorithm. We begin with a section describing the underlying assumptions, and then give a high-level overview of the algorithm. The actual pseudocode appears in Section~\ref{sec:pseudo}. We analyze the the algorithm and its energy complexity in Section~\ref{sec:analysis}.

\subsection{Assumptions}

We assume that we are given a function $f$ describing a deterministic 
protocol to be run on $G$, that is guaranteed to succeed in the OR 
model of message delivery.  We also assume that, prior to attempting to
execute our algorithms indexed by $j$, at least levels $1, 2, \dots, j$ of a
$(R, \alpha, \beta)$ multi-scale clustering have been computed.  Each node knows the ID
of its cluster center, as well as its graphical distance to the cluster center.
(This information is used in UPCAST/DOWNCAST.)

Our algorithms all have the property that each call to one of the defining functions
takes the same number of timesteps, including those spent in recursive calls.  
In some cases, this number of timesteps can be tedious to compute; therefore, we have
taken the liberty to refer to the number of ``simulated'' timesteps, namely the number
of timesteps of the simulated algorithm.  This should be understood as a shorthand only.
Here is a complete list of the state information that each node must store
in its memory.
\begin{itemize}
    \item Shared knowledge of the graph parameters $n, \Delta, D$. Here, $n$ is an upper bound on the number of nodes
    in $G$,  $\Delta$ is an upper bound on the maximum degree of $G$,  $D$ is an upper bound on the diameter of $D$.
    If $\Delta$ and $D$ are not specified, $n-1$ can be used in their place.  We require that all nodes in the network
    start out initialized with the same values of these parameters.
    \item The current timestep, $t$.  Since we are assuming a synchronous network, this value is always the same for all nodes.
    \item The cluster parameters $R_j, \alpha_j, \beta_j$, for $1 \le j \le \ell$.  This knowledge is also shared by all nodes in the
    network.  
    \item A unique ID for our node.  If these are not specified, a random string of $C \log n$ bits may be used; by the birthday paradox, these are distinct with high probability.
    \item For $1 \le j \le \ell$, the ID of our level-$j$ cluster center, as well as our graphical distance to the cluster center.
    This is used for energy-efficient communication by the UPCAST/DOWNCAST subroutines (see~\cite{chang2018energybroadcast}).
    \item A message string, $m$.  Initially null for all nodes. This is used in the NOTIFY protocol.
    \item Any state information stored by our node in the simulated protocol, $f$.  
\end{itemize}

\subsection{Algorithm Overview}

At a high level, the SAF simulation algorithm gets an algorithm $f$ and a time interval $I$, and has the goal of simulating a run of $f$ on $I$ while allowing processors that are far from the action in $f$ to save energy by sleeping until the appropriate time. To this end, SAF simulation uses a multi-scale clustering to estimate the distance to the nodes where the communication is happening, and to set SLEEP and WAKE schedules for the nodes that need to wait. The main idea is that the clusters at level $j$, which are $(R_j, \alpha_j,\beta_j)$ approximately distance-preserving, are equipped to make a decision (SLEEP/WAKE) from their cluster that is pertinent to the next $R_j$ steps of the simulated algorithm. This decision is made for each cluster, collectively by its members, by performing internal cluster operations to decide whether there is activity nearby. However, since their confidence in the prediction of inactivity is only good for $R_j$ simulated timesteps, this would appear to mean that they must wake up and perform cluster operations after every $R_j$ steps of the simulated algorithm. When $R_j$ is small, this would not actually help them save much energy. However, this is where the clusterings at multiple scales come in. At the highest scale, $\ell$, the clusters wake up every 
$R_{\ell}$ steps, so as long as the total simulated time that $f$ needs to run is at most $R_{\ell}\polylog n $, the algorithm is in good shape. The notion of what it means for the action to be far away has been tuned so that nesting property described in Lemma~\ref{lem:nesting} ensures that if a top level cluster has gone to sleep, then all the lower level clusters that are contained within can also afford to sleep for the full time horizon of $R_{\ell}$ steps. If a top level cluster decides to stay awake, then it sets up a recursive call to enable affected lower level clusters to decide how many times they must wake up with the shortened time horizon. Note that the top level cluster staying awake does not mean that an individual node within that cluster must sat awake for $R_{\ell}$ timesteps. Indeed, that would expend far too much energy. Rather, it means that such a node must participate in lower level cluster operations, until it can find a scale at which the computational activity of $f$ is far away from it. Only if a node finds itself to be awake at \emph{all} scales will it participate in the actual simulation of $f$, for the next interval of the smallest scale length, $R_1$.  

\begin{figure}
      \centering
      \includegraphics[scale=0.43]{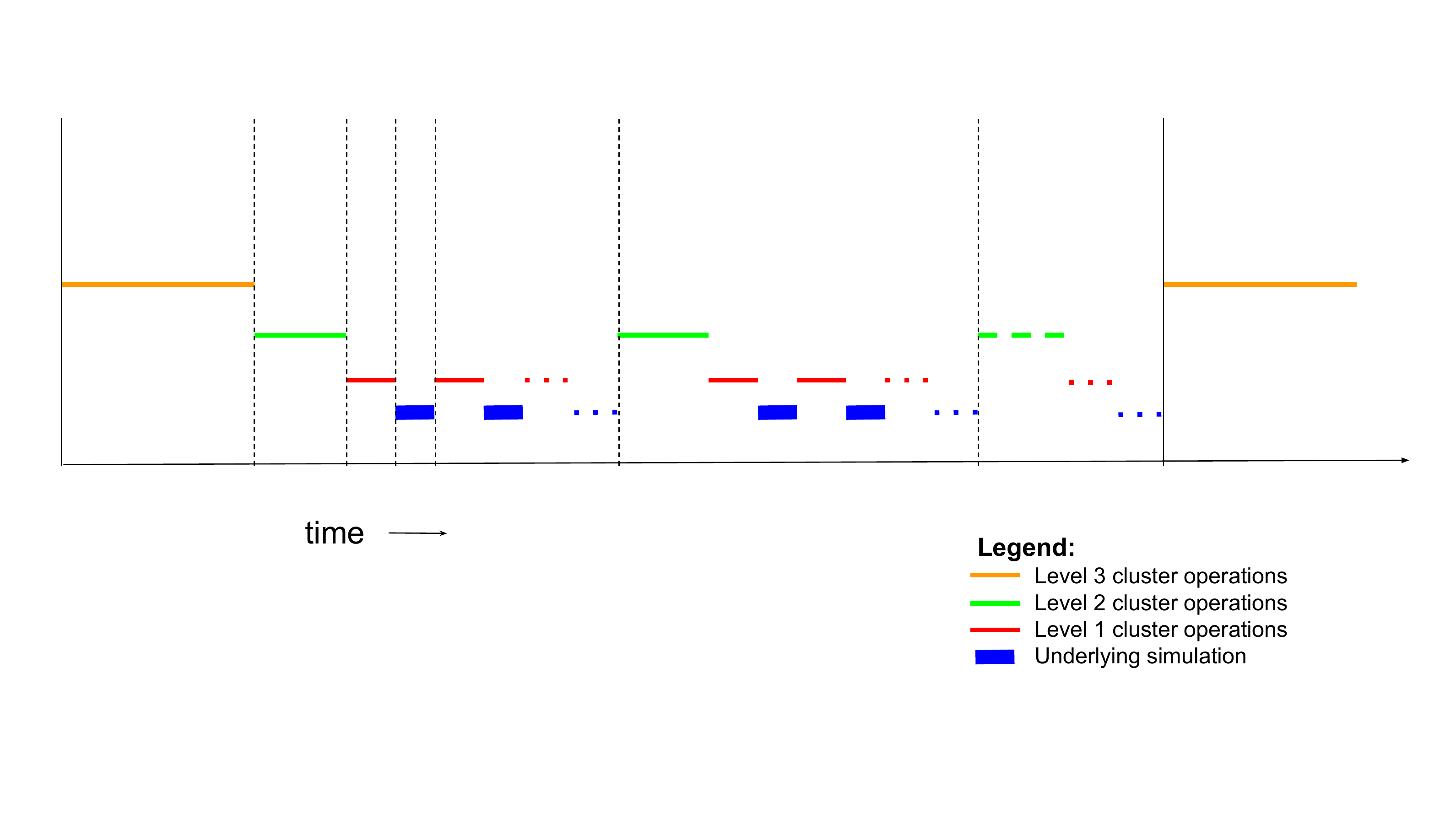}
      \caption{Timeline of how cluster operations at various levels interleave with the underlying simulated algorithm. Each node starts at the top level (here $\ell =3$) and participates in cluster operations until it finds a level at which it can drop out. That level also determines when it will wake up next: the next round of cluster operations at that level. If the node gets to the bottom level of cluster operations without dropping out, then it stays awake for the underlying simulation for the next few simulated timesteps.}
      \label{fig:tributaries}
  \end{figure}

But how to get the multi-scale clustering to begin with? Herein lies the beauty of the process. At the bottom level, an $(R_1, \alpha_1, \beta_1)$ approximate distance-preserving partition can be computed with $\polylog (n)$ energy and  $\polylog (n)$ latency, with all processors staying awake the whole time. (We need $R_1$ to be $\polylog(n)$.) Thereafter, the simulation has a multi-scale clustering to work with, (with $\ell =1$ the first time around, and the partitioning algorithm is itself a function that is suitable for being simulated recursively. 
Thus, each subsequent level of the final multi-scale clustering is bootstrapped off the previously built levels.

\subsection{Pseudocode}\label{sec:pseudo}

Our code is divided into several pieces, which are presented as Algorithms 1 through 4.  The SAF Simulation algorithm is stated in 
a form where it can be applied using an arbitrary multi-scale clustering.  In general, we want to apply it to the multi-scale
clustering we know how to construct for general graphs, namely the ESTCluster algorithm of Miller, Peng, and Xu.  To build this
clustering efficiently, we simulate it using the SAF algorithm inductively.  Finally, to efficiently run BFS from an arbitrary start node,
we apply SAF Simulation to the naive parallel BFS algorithm that has each node listen until the BFS frontier reaches it, then send a 
message to its neighbors to advance the frontier while incrementing its level counter.

\begin{algorithm} \caption{SAF Simulation Algorithm: Adaptively power down receiver to save energy while simulating a given algorithm, $f$.
Assumption: we have a hierarchy of $\alpha, \beta, R$}
\label{alg:SAF}
\begin{algorithmic}
\Procedure{SAF}{$j, I, f$} \Comment{ $j$: current cluster height, $I$: interval of simulated times, $f$: update rule for the simulated algorithm
(on level 0 clusters)}
\If{$j = 0$} \Comment{Bottom level (single node)}
\myState{(Naively) execute $f$ for $|I|$ timesteps.}
\Else 
\myState{Partition $I$ into disjoint subintervals, $I_1, I_2, \dots, I_k$, each of length $R_j$.}
\For{$i \gets 1 \mbox{ to } k$}
\myState{Internally simulate $f$ on $I_i$ assuming no messages received.}
\myState{$m \gets \begin{cases} \mbox{null} & \mbox{if we never sent during the simulation.} \\
\mbox{``Activity nearby!''} & \mbox{if we sent anything during the simulation.}
\end{cases}$}
\If{\Call{Notify}{$j, m, 2 \alpha_j$} returns ``OK TO SLEEP''}
\myState{SLEEP until end of last simulated timestep in $I_i$.}
\myState{Update state based on the above simulation.}
\Else
\myState{\Call{SAF}{$j-1,I_i,f$}}
\EndIf
\EndFor
\EndIf
\EndProcedure
\end{algorithmic}
\end{algorithm}

\begin{algorithm} \caption{Ensure that all nearby clusters (of a given height) wake up when something interesting is happening nearby. }
\begin{algorithmic}
\Procedure{Notify}{$j, m, r$} \Comment{ $j$: current cluster height, $m$: message to send, $r$: transmission radius}
\Comment{ If $m = \mbox{null}$, do not initiate sending, but do repeat any message you hear.} 
\For{$i \gets 1 \mbox{ to }r$}
\myState{\Call{Upcast}{$j$,$m$}} 
\myState{\Call{Downcast}{$j$,$m$}}
\myState{\Call{Intercast}{$j$,$m$}}
\EndFor
\myState{\Call{Upcast}{$j$,$m$}} 
\myState{\Call{Downcast}{$j$,$m$}}
\If{ we received or sent any messages during this call to \Call{Notify}{}}
\myState{\Return{``STAY AWAKE''}}
\Else
\myState{\Return{``OK TO SLEEP''}}
\EndIf
\EndProcedure
\end{algorithmic}
\end{algorithm}

\newcommand{\tmax}{t_{\mathrm{max}}}

\begin{algorithm} \caption{Naively build MPX clusters with radius parameter $R$. }
\label{alg:naive-bmpx}
\begin{algorithmic}
\Procedure{Naively-Build-MPX}{$R$} \Comment{ $R$: radius parameter}
\State{{\tt my_cluster_center}$[j]
\gets ${\tt null}}
\State{Sample a weight $W$ from the exponential distribution with mean $R$.}
\State{$\tmax \gets \left\lceil 3R \log(n)\right\rceil$}
\For{$i \gets 1 \mbox{ to } \tmax$}
\If{{\tt my_cluster_center}$[j] \ne$ {\tt null}}
\State{SEND message ({\tt my_cluster_center}$[j]$,{\tt my_cluster_depth}$[j]$).}
\State{Break out of FOR loop, and SLEEP for remaining $\tmax - i$ timesteps.}
\ElsIf{$i + W \ge \tmax$}
\State{{\tt my_cluster_center}$[j] \gets ${\tt my_ID}}
\State{{\tt my_cluster_depth}$[j] \gets 0$}
\Else
\State{LISTEN this timestep.}
\If{message $(c,d)$ received}
\State{{\tt my_cluster_center[j]} $\gets c$}
\State{{\tt my_cluster_depth}$[j] \gets d+1$}
\EndIf
\EndIf
\EndFor
\EndProcedure
\end{algorithmic}
\end{algorithm}

\begin{algorithm} \caption{Efficiently build a geometric multiscale clustering with parameter $O(\log^2 n)$. }
\label{alg:bmsc}
\begin{algorithmic}
\Procedure{Build-MSC}{} 
\State{$R \gets C \log^2 (n)$}
\For{$j \gets 1 \mbox{ to }\log_R(n)$}
\State{\Call{SAF}{$j-1$, $[0,\left\lceil 3R^j \log(n) \right\rceil ]$, Naively-Build-MPX($R^j$)}}
\EndFor
\EndProcedure
\end{algorithmic}
\end{algorithm}

\section{Analysis}\label{sec:analysis}


\begin{definition}
Suppose $P$ is an $\ell$-level multiscale clustering with parameters $(\mathbf{R}, \boldsymbol\alpha, \boldsymbol\beta)$.
Then, for $0 \le j \le \ell$, we define a \emph{level-$j$ epoch} to be any time interval of the form
\[\{ iR_j + 1, iR_j + 2, \dots, (i+1)R_j\}.\] 
\end{definition}

We will use the following result relating distances from the active set to 
the event of being awake and made to participate in the NOTIFY subroutine.
\begin{lemma}
\label{lem:notify-distance}
For $j \ge 1$, $t = (i+1) R_j$, $v \in V$, if the distance from $v$ to the nearest node that would send in the epoch
ending at time $t$ is at least $(2 \alpha_j+1)(\beta_j R_j+1)$, then, $v$ will not participate in any calls to NOTIFY($j-1$) within this epoch.
\end{lemma}

\begin{proof}
Let $u$ be the nearest node to $v$ that would initiate a SEND in the epoch.
In order to participate in a call to NOTIFY($j-1$), $v$ must receive a ``WAKE UP'' return value from the call to NOTIFY($j$) at the beginning of the epoch.  However, this will not happen because,
by Lemma~\ref{lem:adpp}, $d^*([u],[v]) \ge \left\lfloor \frac{d(u,v)}{\beta_j R_j + 1} \right\rfloor \ge 2 \alpha_j,$
but this means, if $v$ participates in the NOTIFY, it will instead get a return value of ``OK TO SLEEP.''  Here $d$ denotes distance in the underlying graph, and $d^{*}$ denotes distance in the level-$j$ cluster graph.
\end{proof}

Now we are ready to state our main result about our simulation algorithm.

\begin{theorem} \label{thm:SAF-main}
Let $f$ be any randomized radio network protocol in the OR model.
Suppose $P$ is an $\ell$-level multiscale clustering with parameters $(\mathbf{R}, \boldsymbol\alpha, \boldsymbol\beta).$
Then SAF($\ell, I, f)$ has the following properties:
\begin{itemize}
    \item 
    $\Prob{ \mbox {SAF($\ell,I,f)$ succeeds}} \ge \Prob{f \mbox{ succeeds}} - 1/\poly(n).$
    \item The running time of SAF($\ell,I,f)$ is $T \; \polylog(n)$, where $T = \TIME(f)$ is the time complexity of the simulated algorithm.
    \item The energy cost of SAF($\ell,I,f)$ for a vertex $v$ is at most $\mathcal{E} \; \polylog(n)+T \log n/R^\ell$, where $\mathcal{E}$ is the energy cost of the greedy psychic algorithm $\GP(f)$ for vertex $v$, and $T = \TIME(f)$.
\end{itemize}
\end{theorem}


\begin{proof}
First, observe that, after each NOTIFY operation, a node $v$ stays awake if and only if it is 
within $2 \alpha_j$ level-$j$ clusters of the set of active nodes 
for the beginning of the corresponding level-$j$ epoch.
That is, assuming that all of our cluster operations (Upcast/Downcast/Intercast) succeed.
Since, with high probabiltity, this happens at all timesteps, let us discount the $1/\poly(n)$
chance of a failure.

Since the level-$j$ epoch is $R^j$ time units long, every active node in the entire epoch is within distance
$R^j$ of the initial set of active nodes.  By the 
definition of approximately distance-preserving partition, these nodes are all
within distance $\alpha_j$ of the initial active set in the level-$j$ cluster graph.
Since we doubled this radius of notification in the cluster graph, it
follows by Lemma~\ref{lem:nesting} that all nodes in lower-level clusters that wake up during the simulation of
this level-$j$ epoch are already awake at level $j$, and therefore their lower level cluster operations will succeed.

In particular, at level $0$, it follows that any time protocol $f$ has a node Send, and one of its neighbors Listen, 
that both nodes in question will be awake and simulating $f$ at the corresponding timestep.  It follows by induction
that the state of all processors at any simulated time $t$ is 
consistent with $f$ run under the OR model until time $t$.
Hence, the final outcome will also be consistent with $f$ run under the OR model.  

For the latency analysis, we observe that the number of level-$j$ epochs is $T_f/R^j$ where $T_f$ is the running time of $f$.
For each such epoch, we do $O(\alpha_j)$ level-$j$ cluster operations (Upcast/Downcast/Intercast), which each require time
$O(\beta_j R_j \log(n) \log\log(n))$, since the clusters have diameter at most $\beta_j R_j$, and the backoff requires
$\log(\log(n))$ time because each node is, with high probability next to $O(\log n)$ different clusters; the further $O(\log(n))$ factor is to guarantee success with high probability.  Summing over all epochs, we get a total running time 
of $O(T_f \alpha_j \beta_j \log(n) \log\log(n))$ attributable to level-$j$ operations.  Finally, summing over the $ \le \log(n)$ levels,
and using that $\alpha_j, \beta_j = O(\log n)$ for all $j$, we get a total running time of $O(T_f \log^4(n) \log(\log(n))$.

For the energy analysis, we consider the time interval between two consecutive non-SLEEP actions by the Greedy Psychic algorithm, so that OPT is spending one energy.  If this interval has length $L$, then we know that after the first $i$ timesteps of this interval, the distance from $v$ to the nearest active vertex is always at least $L - i$.
Consequently, applying Lemma~\ref{lem:notify-distance}, we know that $v$ can, at most, participate in a subset of the final $(2 \alpha_j + 1)\beta_j$ of the calls to NOTIFY($j$) during this time interval, assuming $j < \ell$. Since each call to NOTIFY costs at most $O(\alpha_j \log n)$ energy per participating vertex, this means our algorithm spends at most
$O(\ell \alpha_j \log n)$ times more energy than the Greedy Psychic algorithm does.  Finally, for the top level, we get  $T/R_{\ell}$ calls to NOTIFY($\ell$) in all, each at a costt of $\alpha_j \log n$, regardless of distances.  Summing these energy costs completes the proof.
\end{proof}

\begin{theorem} \label{thm:BMSC}
The Build-MSC algorithm (Algorithm~\ref{alg:bmsc}) builds a multi-scale clustering
at all scales, with probability $1 - 1/\poly(n)$, 
with total running time $O(D \;\polylog(n))$ and per-node energy usage
$O(\polylog(n))$.
\end{theorem}

\begin{proof}
By Theorem~\ref{thm:SAF-main}, we know that the algorithm succeeds in building each level of
the clustering with essentially the same success probability as Naive-Build-MPX (Algorithm~\ref{alg:naive-bmpx}).
Since the naive algorithm runs in time proportional to the radius parameter, and these radius parameters
form a geometric sequence from $1$ up to $O(D)$, the total running time for all the calls to
Naive-Build-MPX is $O(D \; \polylog(n))$.  Similarly, we can estimate the expected number of active
$j'$-epochs for Naive-Build-MPX($R^j$) as $O(R)$.  
This is because the active set moves out from the cluster centers at unit velocity, 
so once it gets within distance $R^{j'}$ of a vertex, all activity within distance $R^{j'}$ ceases within at most $2 R^{j'}$ timesteps, by the Triangle Inequality.  Thus the total number of active $j$-epochs for a given node is $O(1)$.
Summing over all $0 \le j \le \ell$, we get a total energy use of $\Otilde(\ell)$, which is $\polylog(n)$.
\end{proof}

\section{BFS Revisited}

In this section we apply our methodology to get a polylog energy algorithm for Breadth First Search in radio networks, thus answering an open question from~\cite{chang2020energyBFS}

The algorithm is very simple: we simply simulate the naive BFS algorithm for radio networks in our SAF simulation framework.

\begin{algorithm} 
\caption{Solve BFS from a designated vertex, $v$ in low energy. $D$ is an upper bound on the diameter.}
\label{alg:eff-bfs}
\begin{algorithmic}
\Procedure{Efficient-BFS}{v}
\State{\Call{Build-MSC}{}}
\State{\Call{SAF}{$\ell$, $[0, D]$, Naive-Parallel-BFS($R^j$)}}
\EndProcedure
\end{algorithmic}
\end{algorithm}

\begin{theorem} \label{thm:BFS}
The Efficient-BFS algorithm (Algorithm~\ref{alg:eff-bfs}) computes
the graphical distance to each node from the root vertex, with probability $1 - 1/\poly(n)$.
Its total running time is $O(D \; polylog(n))$ and per-node energy usage is 
$O(\polylog(n))$.
\end{theorem}

\begin{proof}
Theorem~\ref{thm:BMSC} gives us the running time and cost for the cluster formation.
Then we can use Theorem~\ref{thm:SAF-main} to analyze the SAF simulation of the naive BFS algorithm.
The analysis of the number of active epochs here 
is essentially the same as in the proof of Theorem~\ref{thm:BMSC}, since the active set
again passes very quickly through each region that it enters.  We leave the details to the reader.
\end{proof}

\section{Conclusion}

We have shown a new general-purpose methodology for reducing the energy cost of Radio Network algorithms
by collaborating with clusters of nearby nodes at multiple scales to detect when it is safe to shut off
the receiver due to there being no danger of message activity nearby.  Although similar techniques have been 
used in previous work, the precise way in which we create and use our clusters leads, at least in some cases,
to significantly improved results.  In particular, our methodology allows us to easily, and at least in the case
of BFS, significantly improve over known results.

\end{document}